\documentclass[runningheads]{llncs}

\usepackage[utf8]{inputenc}
\usepackage[T1]{fontenc}
\usepackage[english]{babel}

\usepackage{amsmath,amssymb}
\usepackage{siunitx}
\usepackage{graphicx}
\usepackage{subcaption}
\usepackage{xspace}
\usepackage[ruled,linesnumbered]{algorithm2e}
\usepackage{pgfplots}
	\usepgfplotslibrary{colorbrewer}
\pgfplotscreateplotcyclelist{twolinestyles}{solid,dashed}
\pgfplotscreateplotcyclelist{pairedmarks}{%
every mark/.append style={solid,fill=.!80!black},mark=*,mark repeat=10,mark phase=1\\%
every mark/.append style={solid,fill=.!80!black},mark=o,mark repeat=10,mark phase=6\\%
every mark/.append style={solid,fill=.!80!black},mark=square*,mark repeat=10,mark phase=2\\%
every mark/.append style={solid,fill=.!80!black},mark=square,mark repeat=10,mark phase=7\\%
every mark/.append style={solid,fill=.!80!black},mark=triangle*,mark repeat=10,mark phase=3\\%
every mark/.append style={solid,fill=.!80!black},mark=triangle,mark repeat=10,mark phase=8\\%
every mark/.append style={solid,fill=.!80!black},mark=otimes*,mark repeat=10,mark phase=4\\%
every mark/.append style={solid,fill=.!80!black},mark=otimes,mark repeat=10,mark phase=9\\%
every mark/.append style={solid,fill=.!80!black},mark=diamond*,mark repeat=10,mark phase=5\\%
every mark/.append style={solid,fill=.!80!black},mark=diamond,mark repeat=10,mark phase=10\\%
every mark/.append style={solid,fill=.!80!black},mark=pentagon*,mark repeat=10,mark phase=3\\%
every mark/.append style={solid,fill=.!80!black},mark=pentagon,mark repeat=10,mark phase=8\\%
mark=star,mark repeat=10,mark phase=4\\%
mark=star,mark repeat=10,mark phase=9\\%
}
\pgfplotsset{
	cycle list/Paired-10,
	cycle multiindex* list={
    	pairedmarks\nextlist
    	twolinestyles\nextlist
    	Paired-10\nextlist
    },
    every axis plot/.append style={thick},
	xtick=data,
	height=0.4\textwidth,
	width=0.84\textwidth,
	xlabel=Number of operators,
	ylabel=Time (seconds),
	legend pos=outer north east,
	xtick={0,10000,20000,30000,40000,50000},
	scaled x ticks=false
}
\usepackage[hidelinks,breaklinks]{hyperref}

\title{Internal versus external balancing in the evaluation of graph-based number types}
\titlerunning{Internal vs external balancing in the evaluation of graph-based number types}
\author{Hanna Geppert \and Martin Wilhelm}
\authorrunning{H. Geppert \and M. Wilhelm}
\institute{Otto-von-Guericke Universität, Magdeburg, Germany}

\newcommand{\RealAlgebraic}{\texttt{Real\_algebraic}\xspace}
\newcommand{\mpt}{\texttt{mpfr\_t}\xspace}
\newcommand{\mpfr}{MPFR\xspace}
\newcommand{\boost}{Boost\xspace}
\newcommand{\binterval}{\texttt{boost::interval}\xspace}
\newcommand{\gpp}{g\texttt{++}\xspace}
\newcommand{\double}{\texttt{double}\xspace}

\newcommand{\fncompress}{\texttt{compress}\xspace}
\newcommand{\fnraise}{\texttt{raise}\xspace}
\newcommand{\fnsplit}{\texttt{split}\xspace}

\newcommand{\ieee}{IEEE~754\xspace}

\newcommand{\df}{\texttt{def}\xspace}
\newcommand{\bre}{\texttt{bru}\xspace}
\newcommand{\brd}{\texttt{brd}\xspace}
\newcommand{\vdh}{\texttt{ebc}\xspace}
\newcommand{\ebd}{\texttt{ebd}\xspace}
\newcommand{\wbvdh}{\texttt{cmb}\xspace}
\newcommand{\mt}{\texttt{m}\xspace}

\newcommand{\mdf}{\texttt{defm}\xspace}
\newcommand{\mbre}{\texttt{brum}\xspace}
\newcommand{\mbrd}{\texttt{brdm}\xspace}
\newcommand{\mvdh}{\texttt{ebcm}\xspace}
\newcommand{\mebd}{\texttt{ebdm}\xspace}
\newcommand{\mwbvdh}{\texttt{cmbm}\xspace}

\newcommand{\aand}{\textbf{\upshape and}\xspace}

\newcommand{\xapprox}{{\tilde{x}}}
\newcommand{\xval}{x}
\newcommand{\q}{q}
\newcommand{\il}{{i_l}}
\newcommand{\ir}{{i_r}}
\newcommand{\iv}{{i_v}}
\newcommand{\xinc}{\il}
\newcommand{\yinc}{\ir}
\newcommand{\pinc}{\iv}
\newcommand{\cx}{{c_l}}
\newcommand{\cy}{{c_r}}
\newcommand{\wx}{{w_l}}
\newcommand{\wy}{{w_r}}
\newcommand{\wall}{{w_{all}}}
\newcommand{\zx}{{z_l}}
\newcommand{\zy}{{z_r}}
\newcommand{\el}{{e_l}}
\newcommand{\er}{{e_r}}

\newcommand{\dl}{{d_l}}
\newcommand{\dr}{{d_r}}
\newcommand{\fd}{{d_f}}
\newcommand{\fc}{{c_f}}
\newcommand{\cc}{{c}}
\newcommand{\edges}{{\textbf{E}(E)}}

\newcommand{\edelta}{\delta}
\newcommand{\vdelta}{\delta}

\newcommand{\xhigh}{{x_{high}}}
\newcommand{\yhigh}{{y_{high}}}
\newcommand{\xlow}{{x_{low}}}
\newcommand{\ylow}{{y_{low}}}

\newcommand{\pv}{{p_v}}
\newcommand{\Ebal}{{E_{bal}}}
\newcommand{\Elist}{{E_{list}}}

\newcommand{\vroot}{{v_r}}
\newcommand{\vsplitold}{{v_s}}
\newcommand{\vsplitnew}{{v_s'}}

\DeclareMathOperator{\cost}{cost}
\DeclareMathOperator{\constcost}{cost_C}
\DeclareMathOperator{\cp}{cp}
\DeclareMathOperator{\weight}{weight}
\DeclareMathOperator{\wgt}{w}

\DeclareMathOperator{\icost}{cost_i}
\DeclareMathOperator{\vcost}{cost_v}
\DeclareMathOperator{\fcost}{cost_f}
\DeclareMathOperator{\inccost}{cost_r}
\DeclareMathOperator{\pcost}{c_{f}}
\DeclareMathOperator{\inc}{i}
\DeclareMathOperator{\cinc}{c}
\DeclareMathOperator{\paths}{\textbf{P}}
\DeclareMathOperator{\pth}{path}

\newcommand{\vinc}{{\inc(v)}}
\newcommand{\lcinc}{{\cinc(\el)}}
\newcommand{\rcinc}{{\cinc(\er)}}

\begin{document}

\maketitle

\begin{abstract}
Number types for exact computation are usually based on directed acyclic graphs. A poor graph structure can impair the efficency of their evaluation. In such cases the performance of a number type can be drastically improved by restructuring the graph or by internally balancing error bounds with respect to the graph's structure. We compare advantages and disadvantages of these two concepts both theoretically and experimentally.
\end{abstract}

\section{Introduction}
\label{sec:introduction}

Inexact computation causes many problems when algorithms are implemented, ranging from slightly wrong results to crashes or invalid program states. This is especially prevalent in the field of computational geometry, where real number computations and combinatorical properties intertwine~\cite{schirra2000}.
In consequence, various exact number types have been developed~\cite{mehlhorn89,moerig10,yu2010}. It is an ongoing challenge to make these number types sufficiently efficient to be an acceptable alternative to floating-point primitives in practical applications. 
Number types based on the Exact Computation Paradigm recompute the value of complex expressions if the currently stored error bound is not sufficient for an exact decision~\cite{yap1997}. Hence, they store the computation history of a value in a directed acyclic graph, which we call an \emph{expression dag}. The structure of the stored graph is then determined by the order in which the program executes the operations. It lies in the nature of iterative programming that values are often computed step by step, resulting in list-like graph structures. 

Re-evaluating expressions in an unbalanced graph is more expensive than in a balanced one~\cite{vanderHoeven2006,wilhelm17}. 
We discuss two general approaches on reducing the impact of graph structure on the evaluation time.
Prior to the evaluation, the expression dag can be restructured. Originally proposed by Yap~\cite{yap1997}, restructuring methods with varying degrees of invasiveness were developed~\cite{wilhelm17,wilhelm18}. Root-free expression trees can be restructured to reach optimal depth as shown by Brent~\cite{brent74}. In Section~\ref{ssc:restructuring} we introduce a weighted version of Brent's algorithm applied on maximal subtrees inside an expression dag.
Besides restructuring, which can be considered `external' with respect to the evaluation process, we can make `internal' adjustments during the evaluation to compensate for bad structure. Error bounds occuring during an evaluation can be balanced to better reflect the structure of the graph~\cite{vanderHoeven2006}. Doing so requires a switch from an integer to a floating-point error bound representation, leading to numerical issues that need to be taken into consideration~\cite{monniaux2008,wilhelm18error}. In Section~\ref{ssc:errorbalancing} we show how error bounds can be balanced optimally in both the serial and the parallel case and compare several heuristics. Finally, in Section~\ref{sec:experiments} we experimentally highlight strengths and weaknesses of each approach. 

\section{Concepts}
\label{sec:theory}
An \emph{expression dag} is a rooted ordered directed acyclic graph in which each node is either a floating-point number, a unary operation ($\sqrt[d]{\phantom{x}}$,$-$) with one child, or a binary operation ($+$,$-$,$*$,$/$) with two children. We call an expression dag $E'$, whose root is part of another expression dag $E$ a \emph{subexpression of $E$}. We write $v\in E$ to indicate that $v$ is a node in $E$ and we write $|E|$ to represent the number of operator nodes in $E$.
In an \emph{accuracy-driven evaluation} the goal is to evaluate the root node of an expression dag with absolute accuracy $\q$, i.e., to compute an approximation $\xapprox$ for the value $\xval$ of the represented expression, such that $|\xapprox-\xval| \leq 2^{\q}$ (cf.~\cite{yap1997}). To reach this goal, sufficiently small error bounds for the (up to two) child nodes and for the operation error are set and matching approximations are computed recursively for the children. 
Let $v\in E$ be a node with outgoing edges $\el$ to the left and $\er$ to the right child. Let $\inc(\el),\inc(\er)$ be the increase in accuracy for the left and the right child of~$v$ and $\inc(v)$ be the increase in accuracy for the operation (i.e.\ the increase in precision) at $v$. Depending on the operation in $v$ we assign constants $\lcinc,\rcinc$ to its outgoing edges as depicted in Table~\ref{tbl:operationconstants}. If the node $v$ is known from the context, we shortly write $\iv,\il,\ir$ for the accuracy increases at $v,\el,\er$ and $\cx,\cy$ for the respective constants. To guarantee an accuracy of $q$ at $v$, the choice of $\iv,\il,\ir$ must satisfy the inequality
\begin{align}
2^{\q+\iv}+\cx 2^{\q+\il}+\cy 2^{\q+\ir} \leq 2^{\q} \quad\textrm{ or, equivalently, }\quad 2^\iv+\cx 2^\il+\cy 2^\ir\leq 1 
\label{eq:accuracycondition}
\end{align}
Aside from this condition, the choice of $\iv,\il,\ir$ is arbitrary and usually done by 
a symmetric distribution of the error. In the exact number type \RealAlgebraic they are chosen such that $\cx 2^\il\leq 0.25$, $\cy 2^\ir\leq0.25$ and $2^\iv = 0.5$ (and adjusted accordingly for one or zero children). Let the \emph{depth} of a node $v$ in an expression dag be the length of the longest path from the root to $v$. In general, the precision~$\pv$ needed to evaluate a node $v$ increases linearly with the depth of the node due to the steady increase through $\il,\ir$. The approximated value of each node is stored in a multiple-precision floating-point type (bigfloat). The cost of evaluating a node is largely dominated by the cost of the bigfloat operation, which is linear in $|\pv|$ in case of addition and subtraction and linear up to a logarithmic factor in case of multiplication, division and roots. So the precision~$\pv$ is a good indicator for the total evaluation cost of a node (except for negations).
\begin{table}[bth]
\caption{Operation-dependent constants $\cinc(\el)$ and $\cinc(\er)$ for an accuracy-driven evaluation in \RealAlgebraic, with $\xhigh,\yhigh$ upper bounds and $\xlow,\ylow$ lower bounds on the child values.}
\label{tbl:operationconstants}
\centering
\newcommand{\fcolh}[1]{\makebox[0.5cm][c]{#1}}
\newcommand{\colh}[1]{\raisebox{1ex}{\makebox[2cm][c]{#1}}}
\renewcommand{\arraystretch}{1.4}
\vspace{0.2cm}
\begin{tabular}{c|ccccc}
\fcolh{} & \colh{negation} & \colh{add./sub.} & \colh{multipl.} & \colh{division} & \colh{$d$-th root} \\[-1ex]
\hline 
$\cinc(\el)$ & $1$ & $1$ & $\yhigh$ & $\frac{1}{\ylow}$ & $\frac{1}{d}(\xlow)^{\frac{1-d}{d}}$ \\ 
$\cinc(\er)$ & $0$ & $1$ & $\xhigh$ & $\frac{1}{\ylow^2}$ & $0$ \\ 
\end{tabular} 
\end{table}
Let $E$ be an expression dag. We define the \emph{cost of a node $v\in E$} to be $|\pv|$ and the \emph{cost of $E$}, denoted by $\cost(E)$, as the sum of the cost of all nodes in $E$. We set the \emph{depth of $E$} to the maximum depth of all nodes in $E$. Let $\Elist$ be a list-like expression dag, i.e., an expression dag with depth $\Theta(n)$, where $n$ is the number of its nodes and let $\Ebal$ be a balanced expression dag, i.e., an expression dag with depth $\Theta(\log(n))$. Since the precision increases linearly with the depth, we have $\cost(\Elist)=\Theta(n^2)$ and $\cost(\Ebal)=\Theta(n\log(n))$, assuming that the operation constants can be bounded (cf.~\cite{wilhelm17}).
In a parallel environment the cost of the evaluation is driven 
 by dependencies between the nodes. For an expression dag $E$ with $n$ nodes let the \emph{cost of a path} in $E$ be the sum of the cost of the nodes along the path. Let $\cp(E)$ be a path in $E$ with the highest cost. We call $\cp(E)$ a \emph{critical path} in $E$. Then the cost of evaluating $E$ in parallel is $\Theta(\cost(\cp(E)))$ with $O(n)$ processors. Let $\Elist,\Ebal$ be defined as before. Then obviously $\cost(\cp(\Elist))=\Theta(n)$ and $\cost(\cp(\Ebal))=\Theta(\log n)$ (cf.~\cite{wilhelm18}). So in both the serial and the parallel case, balanced graph structures are superior. 

\subsection{Graph Restructuring}
\label{ssc:restructuring}

By definition, exact number types that use accuracy-driven evaluation act lazy, i.e., expressions are not evaluated until a decision needs to be made. Before their first evaluation, underlying graph structures are lightweight and can be changed at low cost. Therefore graph restructuring algorithms ideally take place when the first decision is demanded. While it is not impossible to restructure graphs that have already been evaluated, it comes with several downsides. Since subexpressions will change during restructuring, all approximations and error bounds associated with these subexpressions are lost, although they could be reused in later evaluations. Since stored data may depend on data in subexpressions, the internal state of the whole expression dag may be invalidated. Those effects can make restructuring expensive if many decisions are requested without significant changes to the graph in between. 
Let $E$ be an expression dag. We call a connected, rooted subgraph of $E$ an \emph{operator tree} if it consists solely of operator nodes, does not contain root operations and does not contain nodes with two or more parents (not necessarily in $E$), except for its root. We restructure each maximal operator tree in $E$ according to a weighted version of Brent's algorithm.
Let $T$ be an operator tree in $E$. We call the children of the leaves of $T$ the \emph{operands} of $T$ and associate a positive weight with each of those operands. We define a weight function, such that for each node $v\in T$ the weight of $v$ is greater or equal than the weight of its children. The main difference between the original algorithm and the weighted variation lies in the choice of the split node. We give a brief outline of the algorithm.

The algorithm builds upon two operations, \fncompress and \fnraise. The operation \fncompress takes an expression tree $E$ and returns an expression tree of the form $F/G$ and \fnraise takes an expression tree $E$ and a subtree $X$ and returns an expression tree of the form $(AX+B)/(CX+D)$, where $A,B,C,D,F,G$ are division-free expression trees with logarithmic depth.

\begin{algorithm}[!ht]
\SetKwProg{Fn}{Function}{:}{}
\SetKwFunction{Compress}{compress}
\SetKwFunction{Raise}{raise}
\SetKwFunction{Split}{split}
 \Fn{\Compress{$R$}}{
	\If{$R$ is not an operand}{
 	  $X$ = \Split{$R$,$\frac{1}{2} \weight(R)$};\\
	  let $X_1, X_2$ be the children of $X$;\\
	  \Compress{$X_1$}; \Compress{$X_2$};	\Raise{$R$,$X$};\\
 	  substitute $X$ in $R$;
 	}
 }
 \DontPrintSemicolon\;
 \Fn{\Raise{$R$,$X$}}{
 	\If{$R\neq X$}{
	  $Y$ = \Split{$R$,$\frac{1}{2} (\weight(R) + \weight(X))$};\\
	  let $Y_1, Y_2$ be the children of $Y$, such that $Y_1$ contains $X$;\\
	  \Raise{$Y_1$,$X$}; \Compress{$Y_2$}; \Raise{$R$,$Y$};\\
	  substitute $Y$ in $R$;
	}
 }
\caption{The operations \fncompress and \fnraise.} 
\label{alg:compressraise}
\end{algorithm}

Let $\vroot$ be the root node of $T$. We choose $\vsplitold$ as a node with maximal weight in $T$ such that both children have either weight $<\frac{1}{2}\weight(\vroot)$ or are operands. Note that this implies $\weight(\vsplitold)\geq \frac{1}{2} \weight(\vroot)$. We then recursively call \fncompress on $\vsplitold$ and raise $\vsplitold$ to the root by repeating the following steps:
\begin{enumerate}
\item Search for a new split node $\vsplitnew$ on the path from $\vroot$ to $\vsplitold$ that splits at a weight of $\frac{1}{2}(\weight(\vroot)+\weight(\vsplitold))$.
\item Recursively raise $\vsplitnew$ to $\vroot$ and $\vsplitold$ to the respective child node in $\vsplitnew$.
\item Substitute $\vsplitnew$ and its children into $\vroot$ by incorporating the operation at $\vsplitnew$.
\end{enumerate}
Let $R$ be the expression at $\vroot$, let $Y$ be the expression at $\vsplitnew$ and let $X$ be the expression at $\vsplitold$. After the second step, $R=\frac{A'Y+B'}{C'Y+D'}$ and we have $Y=Y_L\circ Y_R$ with $Y_L=\frac{A''X+B''}{C''X+D''}$ and $Y_R=F''/G''$ or vice versa. Substituting $Y$ (with respect to the operation $\circ$ at $Y$) then gives the desired $R=\frac{AX+B}{CX+D}$. Substituting $X=F'/G'$ finally leads to a balanced expression of the form $R=F/G$.
\begin{algorithm}[!ht]
\SetKwProg{Fn}{Function}{:}{}
\SetKwFunction{Split}{split}
 \Fn{\Split{$X$,$w$}}{
 	\uIf{$X.left$ is not operand \aand $\weight(X.left) \geq w$ \aand $\weight(X.left) \geq \weight(X.right)$ }{
 		\Return \Split{$X$.left,$w$};
 	}
 	\uElseIf{$X.right$ is not operand \aand $\weight(X.right) \geq w$}{
 		\Return \Split{$X$.right,$w$};
 	}
 	\uElse{
 		\Return $X$;
 	}
 }
\caption{The \fnsplit operation.} 
\label{alg:split}
\end{algorithm}

The new \fnsplit operation is shown in Algorithm~\ref{alg:split}. If unit weight is chosen, there will never be an operand that does not satisfy the split condition. If furthermore the weight function is chosen as the number of operands in a subtree, satisfying the split condition implies having a bigger weight than the sibling. Therefore the algorithm is identical to Brent's original algorithm applied to subtrees of the expression dag and guarantees logarithmic depth for the new operator tree. 
Regarding the overall expression dag, nodes which contain root operations, have more than one parent or have been evaluated before are treated equally to the other operands in this case and therefore act as `blocking nodes' for the balancing process. 
Let $k$ be the number of these blocking nodes in $E$. If the number of incoming edges for each blocking node is bounded by a constant, the depth of $E$ after applying the algorithm to each operator tree is in $O(k\log(\frac{n}{k}))$. This depth can be reduced by applying appropriate weights to the blocking nodes.
From a conceptual perspective, a sensible choice for the weight of an operand (as well as for the weights of the inner nodes) would be the number of operator nodes in the subexpression rooted at the operand. Note that we are actually interested in the number of bigfloat operations. However, it is very expensive to compute the number of descendants for a node in a DAG, since one has to deal with duplicates~\cite{borassi2016}. Ignoring duplicates, we could choose the number of operators we would get by expanding the DAG to a tree. While computable in linear time, the number of operators can get exponential (cf.~\cite{vanderHoeven2006,wilhelm17}) and therefore we cannot store the exact weight in an integer data type anymore. There are ways of managing such weights, as we discuss in Section~\ref{ssc:errorbalancing}, but they are imprecise and less efficient than relying on primitives.
Both weight functions behave identical to the unit weight case when there are no blocking nodes present. If there are blocking nodes on the other hand, these nodes get weighted accordingly and expensive nodes are risen to the top of the operator tree. The depth after restructuring for $k$ blocking nodes therefore becomes $O(k+\log n)$.

The weight functions described above are optimal, but hard to compute. Let the weight of both operators and operands be the depth of the subexpression rooted at the operand or operator in the underlying expression dag. Then the algorithm subsequently reduces the length of the longest paths in the expression dag. Note that this strategy does not necessarily lead to an optimal result. Nevertheless, computing the depth of a subexpression in an expression dag can be done fast and the depth can be represented efficiently. Therefore this strategy might prove to be a good heuristic to combine advantages of the unit weight algorithm and the weighted approach.

\subsection{Error Bound Balancing}
\label{ssc:errorbalancing}

As described at the start of this section, the additional cost of unbalanced graph structures originates in the increase in accuracy associated with each node. A more careful choice of $\pinc,\xinc,\yinc$ in (\ref{eq:accuracycondition}) may compensate for an unfavorable structure. If set correctly, linear depth still only leads to a logarithmic increase in accuracy aside from operation constants~\cite{vanderHoeven2006}.

An increase in accuracy at an operator node only affects the operation itself. An increase in accuracy for a child node affects all operations in the subexpression of the child. We associate a non-negative \emph{weight} $\wgt(e)$ with each edge $e$ in an expression dag $E$, representing the impact a change in $\inc(e)$ has on the total cost of $E$. For a node $v\in E$ with outgoing edges $\el,\er$ let $\wx=\wgt(\el)$ and $\wy=\wgt(\er)$.
We then say that, for an evaluation to accuracy $q$, the cost induced on $E$ by the choice of parameters in $v$ is given by
\begin{align}
\label{eq:costerrorboundbalancing}
\icost(v) = -(q + \pinc + \wx\xinc + \wy\yinc)
\end{align}
whereas $\cost(E) = \sum_{v\in E}\icost(v)$. To minimize the total cost we want to minimize the cost induced by each node while maintaining the condition in (\ref{eq:accuracycondition}).
Let $\zx=\cx2^\xinc$, $\zy=\cy2^\yinc$ and let $\wall = 1+\wx+\wy$. With an optimal choice of the parameters, (\ref{eq:accuracycondition}) is an equality and we have $\pinc=\log(1-\zx-\zy)$. Substituting $\pinc$ into (\ref{eq:costerrorboundbalancing}) and setting $\frac{\partial}{\partial \xinc}\icost(v) = \frac{\partial}{\partial \yinc}\icost(v) = 0$ we get
\begin{align}
(1+\wx)\zx+\wx\zy-\wx&=0 \label{eq:partialone}\\
(1+\wy)\zy+\wy\zx-\wy&=0 \label{eq:partialtwo}
\end{align}
leading to $\zx = \frac{\wx}{\wall}$ and $\zy = \frac{\wy}{\wall}$. Resubstituting $\zx$ and $\zy$, the optimal choice of the parameters for error bound distribution inside a node is
\begin{align}
\label{eq:parameterchoice}
\xinc &= \log(\wx) - \log(\wall) - \log(\cx) \nonumber\\ 
\yinc &= \log(\wy) - \log(\wall) - \log(\cy) \\
\pinc &= -\log(\wall)\nonumber
\end{align}

We can show that this parameter choice makes the cost of the evaluation to some degree independent of the structure of the graph. For a node $v\in E$ we denote the set of paths between the root node of $E$ and $v$ by $\paths(v)$. For a path $P\in\paths(v)$ we write $e\in P$ to indicate that $e$ is an edge along $P$.
The precision requested at $v$ along $P$ can be expressed as \(\inccost(P) = \vcost(P)+\fcost(P)\) where 
\[
\textstyle\vcost(P)=-\sum_{e\in P}(\inc(e)+\log(\cinc(e)))-\vinc \quad\textrm{ and }\quad \fcost(P)=\sum_{e\in P}\log(\cinc(e))
\] 
denote the \emph{variable cost} induced by the choice of $\xinc,\yinc,\pinc$ and the \emph{fixed cost} induced by the operation constants along the path.
\begin{theorem}
\label{thm:errorboundcost}
Let $E$ be an expression dag consisting of $n$ unevaluated operator nodes. Then the cost of evaluating $E$ with accuracy $\q\leq 0$ and with an optimal choice of parameters is \[\textstyle\cost(E) = n\log(n) + \sum_{v\in E}\log\left(\sum_{P\in\paths(v)}2^{\fcost(P)}\right) - n\q\]
\end{theorem}
\begin{proof}
We define weights for each node $v$ and each edge $e$ in $E$ with respect to~(\ref{eq:costerrorboundbalancing}). Let
\(\textstyle\pcost(v) = \sum_{P\in\paths(v)}2^{\fcost(P)}\textrm{ and }\pcost(v,e)=\sum_{P_e\in\paths(v),e\in P_e}2^{\fcost(P_e)}\).
Then we set $\wgt(v) = \wall = 1+\wx+\wy$ if $v$ is an operator node and $\wgt(v)=0$ otherwise. For an edge $e$ leading to $v$ we set 
\begin{align}
\label{eq:optimalweight}
\wgt(e) = \frac{\pcost(v,e)}{\pcost(v)}\wgt(v) = \frac{\sum_{P_e\in\paths(v),e\in P_e}2^{\fcost(P_e)}}{\sum_{P\in\paths(v)}2^{\fcost(P)}}\wgt(v)
\end{align}

We show that choosing the parameters as in~(\ref{eq:parameterchoice}) with this weight function is optimal and that it leads to the desired total evaluation cost.
For a node $v\in E$ let $P\in\paths(v)$ be any path to $v$ of the form $P=(v_0,e_0,...,v_k,e_k,v_{k+1}=v)$, then
\begin{align}
&\inccost(P) = \vcost(P) + \fcost(P) \nonumber\\
		 &\textstyle\quad= -\sum_{e\in P}(\inc(e)+\log(\cinc(e)))-\vinc + \fcost(P) \nonumber\\
		 &\textstyle\quad= -\sum_{j=0}^k(\log(\wgt(e_j))-\log(\wgt(v_j))+\log(\wgt(v)) + \fcost(P) \nonumber\\
		 &\textstyle\quad= \log(\wgt(v_0)) - \sum_{j=0}^k\left(\log(\pcost(v_j))+\log(\cinc(e_j))-\log(\pcost(v_{j+1})\right) + \fcost(P) \nonumber\\
		 &\textstyle\quad= \log(\wgt(v_0)) - \log(\pcost(v))
\label{eq:pathincrease}
\end{align}
In particular, the precision requested at $v$ along each path is the same.
Assume that the parameter choice is not optimal. For an edge $e$ let $\edelta(e)$ be the difference in $\inc(e)$ between the optimal value and the value resulting from~(\ref{eq:parameterchoice}) with weights as defined in~(\ref{eq:optimalweight}) and let $\vdelta(v)$ be the respective difference in $\inc(v)$ for a node $v$. Due to the optimization that led to~(\ref{eq:parameterchoice}), the slope of $\inc(v)$ is $-\wgt(\el)$ in direction of $\inc(\el)$ and $-\wgt(\er)$ in direction of $\inc(\er)$ when keeping (\ref{eq:accuracycondition}) equal. So the difference in $\inc(v)$ can be bounded through
\[
\textstyle\vdelta(v) \leq -\edelta(\el)\wgt(\el)-\edelta(\er)\wgt(\er) = -\sum_{v'\in E}\left(\edelta(\el)\frac{\pcost(v',\el)}{\pcost(v')}+\edelta(\er)\frac{\pcost(v',\er)}{\pcost(v')}\right)
\]
Denote the difference in cost by preceeding it with $\Delta$ and let $\edges$ be the set of edges in $E$. For our parameter choice, the precision requested at a node $v$ is the same along each path as shown in~(\ref{eq:pathincrease}), so $\Delta\max_{P\in\paths(v)}\inccost(P)=\max_{P\in\paths(v)}\Delta\inccost(P)$. We then get
\begin{align*}
\Delta\cost(E) &= \sum_{v\in E}\max_{P\in\paths(v)}\Delta\inccost(P) \\
			   &= -\sum_{v\in E}\min_{P\in\paths(v)}\sum_{e\in P}\edelta(e)-\sum_{v\in E}\vdelta(v)\\
			   &\geq -\sum_{v\in E}\min_{P\in\paths(v)}\sum_{e\in P}\edelta(e)+\sum_{v\in E}\sum_{e\in \edges}\edelta(e)\frac{\pcost(v,e)}{\pcost(v)}\\
			   &= -\sum_{v\in E}\min_{P\in\paths(v)}\sum_{e\in P}\edelta(e)+\sum_{v\in E}\sum_{P\in\paths(v)}\sum_{e\in P}\edelta(e)\frac{2^{\fcost(P)}}{\pcost(v)}\\
			   &\geq -\sum_{v\in E}\min_{P\in\paths(v)}\sum_{e\in P}\edelta(e)+\sum_{v\in E}\min_{P\in\paths(v)}\sum_{e\in P}\edelta(e) \;=\; 0
\end{align*}
and therefore our parameter choice is optimal. It remains to calculate the total cost for evaluating $E$. Since $\wgt(v_0)=n$ and each path $P\in\paths(v)$ leads to the same requested precision, the desired equation follows directly from (\ref{eq:pathincrease}) with
\begin{align*}
\cost(E) = \sum_{v\in E}\max_{P\in\paths(v)}\inccost(P) - nq = \sum_{v\in E}(\log(n)-\log(\pcost(v)))-nq \tag*{\qed}
\end{align*}
\end{proof}

Choosing the parameters as in (\ref{eq:parameterchoice}) leads to an optimal distribution of error bounds under the assumption that the weights $\wx,\wy$ accurately reflect the impact of an increase in $\xinc,\yinc$ on the total cost. Computing the exact weight shown in (\ref{eq:optimalweight}) is hard since we have to know and to maintain the cost along all paths leading to a node. We discuss several heuristic approaches. From Theorem~\ref{thm:errorboundcost} we can immediately conclude:
\begin{corollary}
\label{cor:treecost}
Let $T$, $|T|=n$, be an expression tree, i.e., an expression dag where each node has at most one parent. Then the optimal weight choice for an edge leading to a node $v$ is the number of operator nodes in the subexpression rooted at $v$ and the cost of an evaluation of $T$ to accuracy $\q\leq0$ is
\[
	\cost(T) = n\log n + \sum_{v\in V}\fcost(\pth(v)) - nq
\]
where $\pth(v)$ denotes the unique path $P\in\paths(v)$.
\qed
\end{corollary}
So a natural choice for the weight of an edge is the number of operator nodes in the respective subexpression of the target node. Then the optimality condition holds for tree-like expression dags but fails when common subexpressions exist. Figure~\ref{fig:graphweights} shows a graph for which the optimal distribution~(\ref{sfig:optimaldist}) differs from the distribution achieved through counting the operators~(\ref{sfig:operatorcount}). In the example the weights for the middle node are $\wx=\wy=1$. Since the lower addition is a common child of the left and the right path, it gets evaluated only once. The optimal weights would therefore be $\wx=\wy=0.5$. When constants are present it may even occur that a common subexpression already needs to be evaluated at a much higher accuracy and therefore the weight can be set close to zero.
\begin{figure}[hbt]
\captionsetup[subfigure]{justification=centering}
\centering
\begin{subfigure}{0.33\linewidth}
\centering
\includegraphics[scale=0.8,page=1]{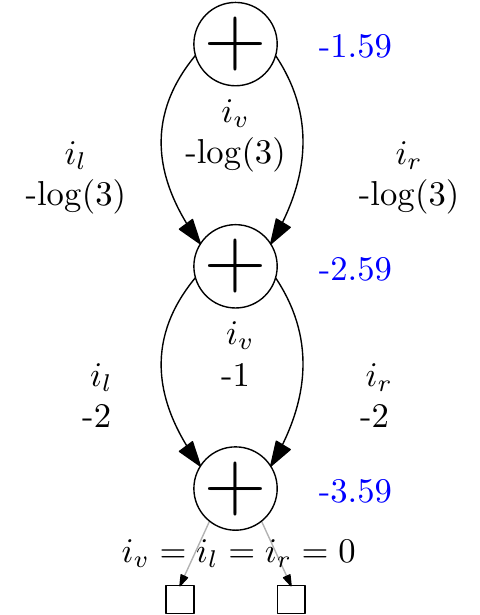}
\subcaption{Optimal distribution}
\label{sfig:optimaldist}
\end{subfigure}%
\begin{subfigure}{0.33\linewidth}
\centering
\includegraphics[scale=0.8,page=2]{balerrormulti}
\subcaption{Single Count}
\label{sfig:operatorcount}
\end{subfigure}
\begin{subfigure}{0.33\linewidth}
\centering
\includegraphics[scale=0.8,page=3]{balerrormulti}
\subcaption{Full Count}
\label{sfig:totaloperators}
\end{subfigure}
\caption{Error bound distribution through different weight functions. The optimal distribution achieves a total cost of $7.77$, while counting the operators with and without removing duplicates has total cost $8.17$ and $8.43$, respectively.}
\label{fig:graphweights}
\end{figure}

Computing the actual number of operators without duplicates in an expression dag is already a difficult task. As in Section~\ref{ssc:restructuring}, we can set the weight of an edge to the number of operators in the subexpression, counting duplicates, in which case we need to deal with a possible exponential increase in weight size. 
This leads to an additional loss in optimality (cf.~Figure~\ref{sfig:totaloperators}), but makes it algorithmically feasible to compute the weights. This approach is largely identical to the one of van der Hoeven, who defined the weights as the number of leaves in the left and right subexpression~\cite{vanderHoeven2006}. Regarding the exponential weight increase, van der Hoeven suggested the use of a floating-point representation. Effectively managing correct floating-point bounds can get expensive. We use a different approach. In the definition of $\pinc,\xinc,\yinc$ the actual value of the weights is never needed. This enables us to store the weight in a logarithmic representation from the start. The downside of this approach is that an exact computation of the weight is not possible even for small values. Note that an overestimation of the weights will never lead us to violate the condition in (\ref{eq:accuracycondition}) and therefore maintains exact computation.
When computing the weights, we need to compute terms of the form $\log(2^a+2^b)$. Let $a\geq b$, then we have $\log(2^a+2^b) = a+\log(1+2^{b-a})$ with $2^{b-a}\leq 1$. An upper bound on the logarithm can be obtained through repeated squaring~\cite{majithia1973}. For $a\gg b$ squaring $1+2^{b-a}$ is numerically unstable. In this case we can approximate the logarithm by linearization near $1$. Then
\begin{align}
\textstyle\log(1+r) \leq \log(1) + r \frac{d}{dx}\log(x)|_1 = \frac{r}{\ln(2)}
\end{align}
and therefore $\log(2^a+2^b) \leq a+\frac{1}{\ln(2)}2^{b-a}$. This approximation works well for a large difference between $a$ and $b$. For small values of $a-b$ we can use repeated squaring. Otherwise we simply set the result to $1$ for $a-b \leq \log(\ln(2))$. One way to efficiently compute an upper bound to the power term is to compute the product $2^{2^{d_1}}\cdots2^{2^{d_k}}$ with $d_{min}\leq d_i\leq 0$ for $1\leq i\leq k$ where $d_1,...,d_k\in\mathbb{Z}$ are the digits set to one in the binary representation of $b-a$. Since the number of possible factors is finite, we can store upper bounds for them in a lookup table.


Error bound balancing does not alter the structure of the expression dag and therefore does not change its parallelizability. The maximum cost of a critical path is reduced from $\Theta(n^2)$ to $\Theta(n\log n)$, but multiple threads cannot be utilized effectively. If an arbitrary number of processors is available, the total cost of the evaluation reduces to the cost of evaluating a critical path. We can therefore choose the error bounds in such a way that the highest cost of a path from the root to a leaf is minimized.
A lower bound on the cost of a critical path $P=(v_0,e_0,...,e_{k-1},v_{k})$ with $k$ operators can be obtained by isolating it, i.e., by assuming that each other edge in the expression dag leads to an operand. Let $\constcost(P)=\sum_{i=0}^k\fcost(\pth(v_i))-kq$ be the cost induced by the constants and the initial accuracy along $P$. Then Corollary~\ref{cor:treecost} gives
\[
\cost(P) = k\log k + \constcost(P)
\]
If $k=n$ the weight choice is already optimal. Let $\Ebal$ be an expression dag that resembles a perfectly balanced tree with depth $k$ and $2^k-1$ operator nodes. Since we do not have common subexpressions, $|\paths(v)| = 1$ for each $v\in\Ebal$ and with~(\ref{eq:pathincrease}) the total cost of any path $P$ in $\Ebal$ is $\cost(P) = k\log(2^k-1)+\constcost(P) = \Theta(k^2)$.
When minimizing the total cost of $\Ebal$, the precision increase $\iv$ at a node $v\in\Ebal$ is weighted against the cost induced in all operators in its subexpression and therefore logarithmic in their number. The cost induced on the critical path, however, depends on the depth of the subexpression. Building upon this observation, the cost of the critical path in $\Ebal$ can be reduced. For a node $v\in\Ebal$ with subexpression depth $j$ and outgoing edges $\el,\er$ we set $\iv=-\log(j)$ and $\il-\log(\cx)=\ir-\log(\cy) = \log(j-1)-\log(j)-1$ (cf.~(\ref{eq:parameterchoice})). Then the cost of the critical path $P$ in $\Ebal$ is
\begin{align}
\cost(P) &= \textstyle-\sum_{j=2}^{k}(-\log(j)+(j-1)(\log(j-1)-\log(j)-1)) + \constcost(P) \nonumber\\
				  &= \textstyle k\log k + \frac{k(k-1)}{2}+\constcost(P)
\label{eq:cpcostbal}
\end{align}
It can be shown that this parameter choice is optimal, aside from taking the operation constants into account. Although not an asymptotic improvement, the cost of the critical path was cut nearly in half. In the derivation of the chosen parameters, we made use of the symmetry of the expression. In general it is hard to compute the optimal parameters for minimizing the critical path.
Let $v$ be the root node of an expression dag $X$ with outgoing edges $\el,\er$ where the left subexpression $L$ has depth $\dl\geq 1$ and the right subexpression $R$ has depth $\dr\geq 1$. In an optimal parameter choice we have 
\begin{align}
\label{eq:cpequality}
\cost(\cp(L)) - \dl\il = \cost(\cp(R)) - \dr\ir
\end{align}
Otherwise, $\il$ or $\ir$ could be decreased without increasing the cost of the critical path of $X$ and $\iv$ could be increased, reducing its cost.
Let $\fd=\frac{\dl}{\dr}$, let $\fc=\frac{\cost(\cp(L))-\cost(\cp(R))}{\dr}$ and let $\cc=2^\fc$. Then $\ir = \fd\il + \fc$ and with (\ref{eq:accuracycondition}) and $z=2^\il$ we get $\iv = \log(1-z-\cc z^\fd)$.
Due to (\ref{eq:cpequality}) there is a critical path through $\el$ and therefore
\[
\cost(\cp(X)) = \cost(\cp(L))-\dl\il-\iv 
\]
Substituting $\iv$ and forming the derivative with respect to $\il$ we get
\begin{align}
\label{eq:cpoptimal}
\frac{-z-\cc\fd z^\fd}{1-z-\cc z^\fd} - \dl = 0 \quad\Longleftrightarrow\quad \cc\frac{\dr-1}{\dr}z^\fd+\frac{\dl-1}{\dl}z-1 = 0
\end{align}
Solving this equation yields an optimal choice for $\il$ (and hence with (\ref{eq:cpequality}) and (\ref{eq:accuracycondition}) for $\ir$ and $\iv$). Note that for $\fd = 1$, $\fc = 0$ and $d_l=d_r$ we get the parameters used for $\Ebal$. Unfortunately, there is no closed form for the solution of (\ref{eq:cpoptimal}) for arbitrary $\fd$. Thus, for an implementation a numerical or a heuristic approach is needed.
\begin{figure}[bt]
\captionsetup[subfigure]{justification=centering}
\centering
\begin{subfigure}{0.33\linewidth}
\centering
\includegraphics[scale=0.8,page=1]{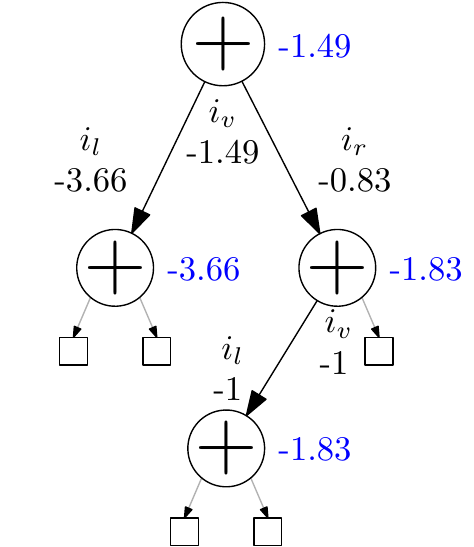}
\subcaption{Optimal path cost}
\label{sfig:optimalcp}
\end{subfigure}%
\begin{subfigure}{0.33\linewidth}
\centering
\includegraphics[scale=0.8,page=2]{balerrorcp}
\subcaption{Optimal total cost}
\label{sfig:optimaltotal}
\end{subfigure}
\begin{subfigure}{0.33\linewidth}
\centering
\includegraphics[scale=0.8,page=3]{balerrorcp}
\subcaption{Depth heuristic}
\label{sfig:depthheuristic}
\end{subfigure}
\caption{Error bound distribution for a graph with two paths of different lengths. In the optimal case, both paths have cost $5.15$. When minimizing total cost, the cost of the critical path is $6$, which gets reduced to $5.74$ with the depth heuristic.}
\label{fig:graphweightscp}
\end{figure}
The cost induced by operation constants and the initial accuracy usually increases with a higher depth. So it is plausible to assume for a node $v$ that the child with the higher subexpression depth will contain a more expensive path in the evaluation, if the difference in accuracy increase at $v$ is relatively small. We can use this observation in the following heuristic. We set
\begin{align}
\iv &= \ir = -\log(\dl+1)-1,\quad \il = \log(\dl)-\log(\dl+1), &\textrm{ if }\; \dl > \dr\nonumber \\
\iv &= \il = -\log(\dr+1)-1,\quad \ir = \log(\dr)-\log(\dr+1), &\textrm{ if }\; \dl < \dr \\
\iv &= -\log(\dl+1),\quad \il = \ir = \log(\dl)-\log(\dl+1)-1, &\textrm{ if }\; \dl = \dr \nonumber
\end{align}
Figure~\ref{fig:graphweightscp} shows an example for the differences between the critical path optimization, total cost optimization and the depth heuristic. The heuristic reduces the weight of the critical path compared to the previous strategies.

\section{Experiments}
\label{sec:experiments}

We present experiments to underline differences between restructuring (Section~\ref{ssc:restructuring}) and error bound balancing (Section~\ref{ssc:errorbalancing}). For the comparison, the policy-based exact-decisions number type \RealAlgebraic with multithreading is used~\cite{moerig10,wilhelm18techreport}.
We compare several different strategies. In our default configuration for \RealAlgebraic we use \binterval as floating-point filter and \mpt as bigfloat data type. Furthermore we always enable topological evaluation, bottom-up separation bound representation and error representation by exponents~\cite{moerig15macis,wilhelm18error}. We call the default strategy without balancing \df. For restructuring we use the weighted version of Brent's algorithm with unit weights~(\bre) and with setting the weights to the expression depth~(\brd). For error bound balancing we use the weight function counting all operators without removing duplicates~(\vdh) and the depth-based approach for reducing the length of critical paths~(\ebd). We furthermore test combinations of internal and external balancing as described in the respective sections. For every strategy we use a variant with and without multithreading~(\mt).
The experiments are performed on an Intel i7-4700MQ with 16GB RAM under Ubuntu 18.04, using \gpp~7.3.0, \boost~1.62.0 and \mpfr~4.0.1. All data points are averaged over twenty runs if not specified otherwise. All expressions are evaluated to an accuracy of $\q=-10000$.

\subsection{List-like expression dags}
\label{ssc:explistlike}

List-like expression dags with linear depth have quadratic cost (cf.~Section~\ref{sec:theory}). Both restructuring and error bound balancing should reduce the cost significantly in this case. We build an expression dag $\Elist$ by computing $res := res \circ a_i$ in a simple loop starting with $res=a_0$, where $\circ\in\lbrace +,-,*,/\rbrace$ is chosen randomly and uniformly and $a_i$ are operands ($0\leq i\leq n$). For the operands we choose random rationals, i.e., expressions of the form $a_i=d_{i,1}/d_{i,2}$ where $d_{i,j}\neq 0$ are random \double numbers exponentially distributed around $1$. By using exact divisions we assure that the operands have sufficient complexity for our experiments. To prevent them from being affected by restructuring, we assign an additional (external) reference to each operand.
Figure~\ref{fig:listlike} shows the results for evaluating $\Elist$. 
\begin{figure}[bht]
\centering
\begin{tikzpicture}
\begin{axis}[
	ymode=log]

\addplot coordinates { (1000,0.011) (2000,0.026) (3000,0.046) (4000,0.07) (5000,0.097) (6000,0.131) (7000,0.17) (8000,0.214) (9000,0.262) (10000,0.317) (11000,0.38) (12000,0.44) (13000,0.51) (14000,0.59) (15000,0.68) (16000,0.77) (17000,0.87) (18000,0.98) (19000,1.09) (20000,1.22) (21000,1.34) (22000,1.48) (23000,1.62) (24000,1.76) (25000,1.93) (26000,2.11) (27000,2.28) (28000,2.46) (29000,2.66) (30000,2.83) (31000,3.06) (32000,3.27) (33000,3.52) (34000,3.75) (35000,4) (36000,4.25) (37000,4.55) (38000,4.78) (39000,5.06) (40000,5.35) (41000,5.67) (42000,5.95) (43000,6.27) (44000,6.59) (45000,6.89) (46000,7.24) (47000,7.59) (48000,7.93) (49000,8.32) (50000,8.71)};
\addplot coordinates { (1000,0.014) (2000,0.031) (3000,0.05) (4000,0.074) (5000,0.1) (6000,0.133) (7000,0.17) (8000,0.215) (9000,0.264) (10000,0.318) (11000,0.37) (12000,0.44) (13000,0.51) (14000,0.58) (15000,0.67) (16000,0.76) (17000,0.87) (18000,0.97) (19000,1.08) (20000,1.19) (21000,1.33) (22000,1.45) (23000,1.59) (24000,1.75) (25000,1.89) (26000,2.05) (27000,2.26) (28000,2.4) (29000,2.56) (30000,2.76) (31000,3) (32000,3.19) (33000,3.42) (34000,3.67) (35000,3.88) (36000,4.13) (37000,4.38) (38000,4.63) (39000,4.9) (40000,5.18) (41000,5.46) (42000,5.75) (43000,6.04) (44000,6.36) (45000,6.64) (46000,7.03) (47000,7.34) (48000,7.69) (49000,8.11) (50000,8.41)};
\addplot coordinates { (1000,0.014) (2000,0.028) (3000,0.043) (4000,0.06) (5000,0.074) (6000,0.09) (7000,0.107) (8000,0.124) (9000,0.142) (10000,0.158) (11000,0.17) (12000,0.18) (13000,0.2) (14000,0.22) (15000,0.23) (16000,0.25) (17000,0.27) (18000,0.29) (19000,0.3) (20000,0.32) (21000,0.34) (22000,0.36) (23000,0.37) (24000,0.39) (25000,0.41) (26000,0.43) (27000,0.45) (28000,0.46) (29000,0.49) (30000,0.5) (31000,0.52) (32000,0.54) (33000,0.56) (34000,0.58) (35000,0.59) (36000,0.62) (37000,0.63) (38000,0.66) (39000,0.67) (40000,0.69) (41000,0.71) (42000,0.73) (43000,0.75) (44000,0.77) (45000,0.78) (46000,0.8) (47000,0.83) (48000,0.83) (49000,0.85) (50000,0.88)};
\addplot coordinates { (1000,0.009) (2000,0.016) (3000,0.025) (4000,0.035) (5000,0.043) (6000,0.055) (7000,0.067) (8000,0.077) (9000,0.089) (10000,0.1) (11000,0.1) (12000,0.11) (13000,0.12) (14000,0.13) (15000,0.14) (16000,0.15) (17000,0.16) (18000,0.17) (19000,0.18) (20000,0.19) (21000,0.2) (22000,0.22) (23000,0.23) (24000,0.24) (25000,0.25) (26000,0.26) (27000,0.27) (28000,0.29) (29000,0.29) (30000,0.31) (31000,0.32) (32000,0.32) (33000,0.34) (34000,0.35) (35000,0.37) (36000,0.38) (37000,0.39) (38000,0.41) (39000,0.41) (40000,0.42) (41000,0.43) (42000,0.44) (43000,0.45) (44000,0.47) (45000,0.48) (46000,0.49) (47000,0.5) (48000,0.51) (49000,0.52) (50000,0.54)};
\pgfplotsset{cycle list shift=2}
\addplot coordinates { (1000,0.009) (2000,0.02) (3000,0.032) (4000,0.047) (5000,0.063) (6000,0.079) (7000,0.098) (8000,0.114) (9000,0.136) (10000,0.155) (11000,0.17) (12000,0.19) (13000,0.22) (14000,0.25) (15000,0.27) (16000,0.3) (17000,0.33) (18000,0.36) (19000,0.39) (20000,0.42) (21000,0.46) (22000,0.49) (23000,0.53) (24000,0.57) (25000,0.61) (26000,0.66) (27000,0.69) (28000,0.74) (29000,0.78) (30000,0.83) (31000,0.88) (32000,0.93) (33000,0.99) (34000,1.04) (35000,1.09) (36000,1.15) (37000,1.2) (38000,1.27) (39000,1.33) (40000,1.39) (41000,1.46) (42000,1.53) (43000,1.59) (44000,1.65) (45000,1.73) (46000,1.8) (47000,1.89) (48000,1.95) (49000,2.02) (50000,2.12)};
\addplot coordinates { (1000,0.011) (2000,0.024) (3000,0.038) (4000,0.053) (5000,0.068) (6000,0.085) (7000,0.105) (8000,0.124) (9000,0.143) (10000,0.166) (11000,0.18) (12000,0.2) (13000,0.23) (14000,0.25) (15000,0.28) (16000,0.31) (17000,0.34) (18000,0.37) (19000,0.4) (20000,0.43) (21000,0.46) (22000,0.51) (23000,0.54) (24000,0.57) (25000,0.61) (26000,0.66) (27000,0.7) (28000,0.74) (29000,0.79) (30000,0.84) (31000,0.89) (32000,0.94) (33000,0.98) (34000,1.04) (35000,1.09) (36000,1.15) (37000,1.22) (38000,1.27) (39000,1.32) (40000,1.38) (41000,1.46) (42000,1.52) (43000,1.6) (44000,1.66) (45000,1.72) (46000,1.79) (47000,1.85) (48000,1.93) (49000,2) (50000,2.07)};

\legend{\df,\mdf,\bre,\mbre,\vdh,\mvdh}
\end{axis}

\end{tikzpicture}
\caption{Running times on a list-like expression dag. Restructuring reduces times by up to \SI{90}{\percent} for single-threaded and by up to \SI{94}{\percent} for multithreaded evaluation. Error bound balancing reduces the running time by up to \SI{75}{\percent} in both cases.}
\label{fig:listlike}
\end{figure}
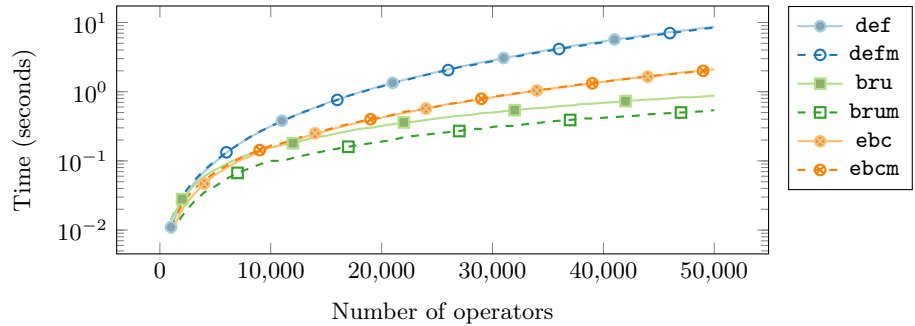
Both balancing methods lead to a significant reduction in running time compared to the default configuration (note the logarithmic scale). For large numbers of operators, restructuring is superior to error bound balancing. While error bound balancing optimizes the variable precision increase, it does not reduce the cost associated with the operation constants. The precision increase due to operation constants affects more nodes in an unbalanced structure than in a balanced one, which gives restructuring an advantage. For small numbers of operators, error bound balancing leads to better results than restructuring, since the cost of evaluating additional operators created through restructuring becomes more relevant.
The structure of $\Elist$ is highly detrimental to efficient parallelization. Consequently, neither the default evaluation nor the error bound balanced evaluation show significant cost reduction when run on multiple processors. With Brent's algorithm a speedup of about $1.7$, i.e., a runtime reduction of about \SI{40}{\percent}, can be observed.
Since $\Elist$ does not contain any common subexpressions or other barriers, the results for other restructuring or error bound balancing strategies are indistinguishable from their counterparts. Interestingly, the evaluation does not benefit from a combination of both balancing strategies. Instead the results closely resemble the results obtained by using only restructuring and even get a bit worse in the multithreaded case. Since through restructuring a perfectly balanced dag is created, the default error bounds are already close to optimal (cf.~Section~\ref{ssc:expbalanced}).

\subsection{Blocking nodes}
\label{ssc:expblocking}

Restructuring gets difficult as soon as `blocking nodes', such as nodes with multiple parents, occur in the expression dag (cf.~Section~\ref{ssc:restructuring}). We repeat the experiment from Section~\ref{ssc:explistlike}, but randomly let about \SI{30}{\percent} of the operator nodes be blocking nodes by adding an additional parent (which is not part of our evaluation). Nodes with such a parent cannot be part of a restructuring process, since the subexpressions associated with them might be used somewhere else and therefore cannot be destroyed.
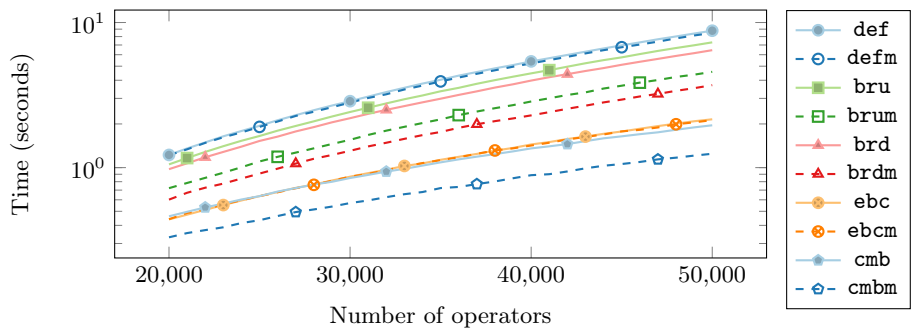
\begin{figure}[hbt]
\centering
\begin{tikzpicture}
\begin{axis}[
	ymode=log]
	
\addplot coordinates { (20000,1.223) (21000,1.356) (22000,1.492) (23000,1.644) (24000,1.788) (25000,1.942) (26000,2.121) (27000,2.298) (28000,2.487) (29000,2.678) (30000,2.87) (31000,3.086) (32000,3.317) (33000,3.55) (34000,3.795) (35000,4.031) (36000,4.305) (37000,4.556) (38000,4.833) (39000,5.091) (40000,5.396) (41000,5.677) (42000,5.977) (43000,6.294) (44000,6.627) (45000,6.959) (46000,7.295) (47000,7.666) (48000,8.017) (49000,8.395) (50000,8.78)};
\addplot coordinates { (20000,1.217) (21000,1.337) (22000,1.481) (23000,1.619) (24000,1.748) (25000,1.909) (26000,2.079) (27000,2.271) (28000,2.427) (29000,2.607) (30000,2.806) (31000,3.02) (32000,3.202) (33000,3.456) (34000,3.676) (35000,3.928) (36000,4.176) (37000,4.433) (38000,4.686) (39000,4.959) (40000,5.228) (41000,5.509) (42000,5.802) (43000,6.126) (44000,6.448) (45000,6.756) (46000,7.049) (47000,7.384) (48000,7.773) (49000,8.142) (50000,8.475)};
\addplot coordinates { (20000,1.053) (21000,1.164) (22000,1.283) (23000,1.405) (24000,1.531) (25000,1.652) (26000,1.798) (27000,1.937) (28000,2.085) (29000,2.25) (30000,2.417) (31000,2.586) (32000,2.77) (33000,2.96) (34000,3.142) (35000,3.359) (36000,3.547) (37000,3.783) (38000,4.008) (39000,4.233) (40000,4.463) (41000,4.696) (42000,4.961) (43000,5.239) (44000,5.506) (45000,5.752) (46000,6.06) (47000,6.372) (48000,6.672) (49000,6.969) (50000,7.3)};
\addplot coordinates { (20000,0.722) (21000,0.784) (22000,0.85) (23000,0.922) (24000,1.006) (25000,1.103) (26000,1.187) (27000,1.269) (28000,1.358) (29000,1.455) (30000,1.555) (31000,1.68) (32000,1.798) (33000,1.911) (34000,2.029) (35000,2.151) (36000,2.301) (37000,2.435) (38000,2.568) (39000,2.696) (40000,2.856) (41000,3.006) (42000,3.167) (43000,3.301) (44000,3.508) (45000,3.669) (46000,3.853) (47000,4.015) (48000,4.216) (49000,4.371) (50000,4.572)};
\addplot coordinates { (20000,0.976) (21000,1.071) (22000,1.179) (23000,1.283) (24000,1.395) (25000,1.527) (26000,1.636) (27000,1.776) (28000,1.903) (29000,2.052) (30000,2.191) (31000,2.347) (32000,2.501) (33000,2.655) (34000,2.818) (35000,2.985) (36000,3.17) (37000,3.355) (38000,3.562) (39000,3.753) (40000,3.981) (41000,4.198) (42000,4.406) (43000,4.634) (44000,4.853) (45000,5.114) (46000,5.363) (47000,5.615) (48000,5.885) (49000,6.15) (50000,6.448)};
\addplot coordinates { (20000,0.602) (21000,0.671) (22000,0.727) (23000,0.776) (24000,0.841) (25000,0.91) (26000,0.984) (27000,1.066) (28000,1.151) (29000,1.222) (30000,1.306) (31000,1.4) (32000,1.484) (33000,1.579) (34000,1.675) (35000,1.761) (36000,1.881) (37000,1.99) (38000,2.091) (39000,2.183) (40000,2.291) (41000,2.412) (42000,2.544) (43000,2.672) (44000,2.815) (45000,2.943) (46000,3.093) (47000,3.22) (48000,3.383) (49000,3.528) (50000,3.702)};
\addplot coordinates { (20000,0.44) (21000,0.474) (22000,0.512) (23000,0.552) (24000,0.594) (25000,0.636) (26000,0.679) (27000,0.725) (28000,0.764) (29000,0.815) (30000,0.866) (31000,0.916) (32000,0.967) (33000,1.027) (34000,1.069) (35000,1.132) (36000,1.187) (37000,1.243) (38000,1.309) (39000,1.367) (40000,1.444) (41000,1.5) (42000,1.575) (43000,1.636) (44000,1.71) (45000,1.774) (46000,1.853) (47000,1.933) (48000,2.01) (49000,2.087) (50000,2.153)};
\addplot coordinates { (20000,0.441) (21000,0.487) (22000,0.515) (23000,0.559) (24000,0.596) (25000,0.636) (26000,0.676) (27000,0.721) (28000,0.761) (29000,0.811) (30000,0.866) (31000,0.905) (32000,0.96) (33000,1.018) (34000,1.075) (35000,1.126) (36000,1.178) (37000,1.245) (38000,1.313) (39000,1.358) (40000,1.42) (41000,1.485) (42000,1.563) (43000,1.628) (44000,1.692) (45000,1.778) (46000,1.821) (47000,1.892) (48000,1.99) (49000,2.051) (50000,2.125)};
\pgfplotsset{cycle list shift=2}
\addplot coordinates { (20000,0.462) (21000,0.496) (22000,0.532) (23000,0.565) (24000,0.603) (25000,0.636) (26000,0.684) (27000,0.723) (28000,0.765) (29000,0.799) (30000,0.848) (31000,0.891) (32000,0.939) (33000,0.981) (34000,1.033) (35000,1.087) (36000,1.14) (37000,1.19) (38000,1.239) (39000,1.294) (40000,1.357) (41000,1.402) (42000,1.453) (43000,1.522) (44000,1.592) (45000,1.643) (46000,1.7) (47000,1.751) (48000,1.824) (49000,1.894) (50000,1.958)};
\addplot coordinates { (20000,0.331) (21000,0.354) (22000,0.371) (23000,0.388) (24000,0.415) (25000,0.434) (26000,0.464) (27000,0.494) (28000,0.515) (29000,0.539) (30000,0.57) (31000,0.594) (32000,0.626) (33000,0.656) (34000,0.679) (35000,0.723) (36000,0.736) (37000,0.77) (38000,0.803) (39000,0.845) (40000,0.887) (41000,0.901) (42000,0.941) (43000,0.989) (44000,1.03) (45000,1.059) (46000,1.094) (47000,1.14) (48000,1.171) (49000,1.21) (50000,1.246)};

\legend{\df,\mdf,\bre,\mbre,\brd,\mbrd,\vdh,\mvdh,\wbvdh,\mwbvdh}
\end{axis}

\end{tikzpicture}
\caption{Running times on a list-like expression dag where \SI{30}{\percent} of the operators have an additional reference. Error bound balancing is not affected by the references, restructuring performs much worse. Combining error bound balancing and restructuring leads to the best results for multithreading.}
\label{fig:blocked}
\end{figure}
Both the default and the internal balancing method are not affected by the change and thus show the same results as before. Restructuring on the other hand performs worse and falls back behind error bound balancing (cf.~Figure~\ref{fig:blocked}). The depth heuristic leads to fewer losses for both total and parallel running time. It reduces the running time by about \SI{10}{\percent} in single-threaded and about \SI{20}{\percent} in multithreaded execution compared to using unit weights.
Combining internal and external balancing combines the advantages of both strategies in this case. Exemplarly, a combination of \brd and \vdh, named \wbvdh, is shown in Figure~\ref{fig:blocked}. For serial evaluation the running time of the combined approach mostly resembles the running time of error bound balancing, getting slightly faster for a large number of operators (about~\SI{9}{\percent} for $N=50000$). In parallel, however, it strongly increases parallelizability leading to a speedup of $1.6$ and a total runtime reduction of up to \SI{85}{\percent} compared to the default strategy.

\subsection{Balanced expression dags}
\label{ssc:expbalanced}

When an expression dag is already balanced, there is not much to gain by either balancing method. In a perfectly balanced expression dag $\Ebal$, restructuring cannot reduce the depth and therefore does not reduce its cost, neither in serial nor in parallel. Brent's algorithm still creates a normal form, which adds additional operations and might even increase the maximum depth. Error bound balancing on the other hand can potentially make a difference.
For a balanced expression dag the total cost is strongly influenced by the operation constants, which is reflected in a high variance when choosing the operators at random. In the experiment shown in Figure~\ref{fig:balanced}, we increase the number of test sets for each data point from $20$ to $50$ and use the same test data for each number type. The single data points lie in a range of about $\pm\SI{20}{\percent}$ of the respective average.
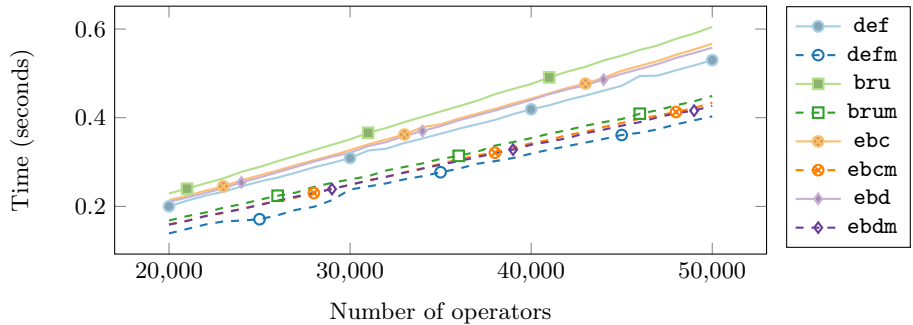
\begin{figure}[tbh]
\centering
\begin{tikzpicture}
\begin{axis}

\addplot coordinates { (20000,0.2) (21000,0.213) (22000,0.224) (23000,0.233) (24000,0.245) (25000,0.256) (26000,0.265) (27000,0.276) (28000,0.288) (29000,0.298) (30000,0.309) (31000,0.326) (32000,0.33) (33000,0.343) (34000,0.353) (35000,0.364) (36000,0.375) (37000,0.385) (38000,0.395) (39000,0.408) (40000,0.419) (41000,0.428) (42000,0.44) (43000,0.45) (44000,0.461) (45000,0.472) (46000,0.494) (47000,0.495) (48000,0.507) (49000,0.518) (50000,0.53)};
\addplot coordinates { (20000,0.139) (21000,0.149) (22000,0.159) (23000,0.166) (24000,0.168) (25000,0.171) (26000,0.18) (27000,0.192) (28000,0.199) (29000,0.214) (30000,0.238) (31000,0.245) (32000,0.252) (33000,0.261) (34000,0.268) (35000,0.277) (36000,0.285) (37000,0.296) (38000,0.302) (39000,0.309) (40000,0.319) (41000,0.328) (42000,0.336) (43000,0.344) (44000,0.352) (45000,0.361) (46000,0.366) (47000,0.374) (48000,0.385) (49000,0.394) (50000,0.403)};
\addplot coordinates { (20000,0.229) (21000,0.24) (22000,0.252) (23000,0.263) (24000,0.278) (25000,0.289) (26000,0.302) (27000,0.314) (28000,0.327) (29000,0.339) (30000,0.351) (31000,0.366) (32000,0.377) (33000,0.39) (34000,0.402) (35000,0.414) (36000,0.426) (37000,0.438) (38000,0.453) (39000,0.465) (40000,0.476) (41000,0.491) (42000,0.504) (43000,0.515) (44000,0.529) (45000,0.54) (46000,0.553) (47000,0.563) (48000,0.578) (49000,0.59) (50000,0.605)};
\addplot coordinates { (20000,0.168) (21000,0.178) (22000,0.186) (23000,0.197) (24000,0.205) (25000,0.215) (26000,0.224) (27000,0.231) (28000,0.244) (29000,0.253) (30000,0.261) (31000,0.268) (32000,0.281) (33000,0.289) (34000,0.297) (35000,0.307) (36000,0.314) (37000,0.323) (38000,0.337) (39000,0.345) (40000,0.354) (41000,0.364) (42000,0.372) (43000,0.382) (44000,0.39) (45000,0.397) (46000,0.409) (47000,0.417) (48000,0.428) (49000,0.437) (50000,0.449)};
\pgfplotsset{cycle list shift=2}
\addplot coordinates { (20000,0.214) (21000,0.224) (22000,0.235) (23000,0.245) (24000,0.259) (25000,0.27) (26000,0.281) (27000,0.293) (28000,0.304) (29000,0.315) (30000,0.327) (31000,0.34) (32000,0.351) (33000,0.362) (34000,0.379) (35000,0.385) (36000,0.398) (37000,0.409) (38000,0.42) (39000,0.432) (40000,0.443) (41000,0.455) (42000,0.467) (43000,0.477) (44000,0.49) (45000,0.506) (46000,0.517) (47000,0.528) (48000,0.542) (49000,0.554) (50000,0.567)};
\addplot coordinates { (20000,0.159) (21000,0.168) (22000,0.178) (23000,0.185) (24000,0.194) (25000,0.204) (26000,0.212) (27000,0.22) (28000,0.23) (29000,0.24) (30000,0.248) (31000,0.259) (32000,0.266) (33000,0.277) (34000,0.286) (35000,0.296) (36000,0.306) (37000,0.315) (38000,0.321) (39000,0.333) (40000,0.342) (41000,0.35) (42000,0.361) (43000,0.368) (44000,0.379) (45000,0.388) (46000,0.398) (47000,0.405) (48000,0.413) (49000,0.422) (50000,0.434)};
\addplot coordinates { (20000,0.21) (21000,0.22) (22000,0.23) (23000,0.241) (24000,0.254) (25000,0.265) (26000,0.277) (27000,0.289) (28000,0.301) (29000,0.312) (30000,0.321) (31000,0.336) (32000,0.345) (33000,0.359) (34000,0.37) (35000,0.381) (36000,0.394) (37000,0.405) (38000,0.416) (39000,0.428) (40000,0.44) (41000,0.453) (42000,0.464) (43000,0.473) (44000,0.486) (45000,0.499) (46000,0.51) (47000,0.521) (48000,0.535) (49000,0.546) (50000,0.558)};
\addplot coordinates { (20000,0.159) (21000,0.167) (22000,0.178) (23000,0.186) (24000,0.194) (25000,0.203) (26000,0.213) (27000,0.222) (28000,0.23) (29000,0.239) (30000,0.249) (31000,0.258) (32000,0.267) (33000,0.276) (34000,0.285) (35000,0.294) (36000,0.302) (37000,0.311) (38000,0.32) (39000,0.328) (40000,0.339) (41000,0.346) (42000,0.354) (43000,0.365) (44000,0.373) (45000,0.382) (46000,0.39) (47000,0.401) (48000,0.41) (49000,0.416) (50000,0.427)};

\legend{\df,\mdf,\bre,\mbre,\vdh,\mvdh,\ebd,\mebd}
\end{axis}

\end{tikzpicture}
\caption{Running times on a perfectly balanced expression dag. All balancing approaches lead to a performance loss. Restructuring increases the running time by about \SI{15}{\percent}, error bound balancing by \SIrange{5}{7}{\percent} in the single-threaded case.}
\label{fig:balanced}
\end{figure}
As expected, restructuring performs worse than the default number type, doubling the depth and replacing each division by, on average, two multiplications. Error bound balancing performs worse than not balancing as well. Neither the total operator count, nor the depth-based strategy have a significant impact on the running time of the bigfloat operations, since the cost decrease per operation is at most logarithmic in the number of operators. For the same reason and due to the limited number of processors, the expected cost reduction between $\mvdh$ and $\mebd$ in the multithreaded case (cf.\ Section~\ref{ssc:errorbalancing}) can not be observed in the experimental data.

\subsection{Common Subexpressions}

Random expression dags, created by randomly applying operations on a forest of operands until it is reduced to a single DAG, tend to be balanced and therefore behave similarly to a perfectly balanced tree. This changes if common subexpressions are involved. With error bound balancing, common subexpressions can be recognized and the error bounds at the parent nodes can be adjusted, such that both request the same accuracies (cf.~Theorem~\ref{thm:errorboundcost}). The two implemented heuristics to some degree take common subexpressions into account, since they contribute the same weight to all of the subexpression's parents. 
We test the behavior of error bound balancing strategies by randomly reusing a certain percentage of subtrees during randomized bottom-up construction of the graph. To avoid zeros, ones, or an exponential explosion of the expression's value we only use additions if two subtrees are identical during construction. While error bound balancing still cannot outperform the default strategy due to the balanced nature, it moves on par with it. If \SI{5}{\percent} of the operations have more than one parent, the error bound balancing strategies improve the single-threaded running time by about \SIrange{1}{5}{\percent}. In a parallel environment, it still performs worse with \mebd being slightly superior to \mvdh.
\begin{figure}[ht]
\centering
\begin{tikzpicture}
\begin{axis}

\addplot coordinates { (20000,0.054) (21000,0.059) (22000,0.064) (23000,0.069) (24000,0.074) (25000,0.079) (26000,0.085) (27000,0.091) (28000,0.097) (29000,0.103) (30000,0.11) (31000,0.116) (32000,0.123) (33000,0.13) (34000,0.137) (35000,0.144) (36000,0.152) (37000,0.16) (38000,0.169) (39000,0.176) (40000,0.184) (41000,0.193) (42000,0.206) (43000,0.212) (44000,0.221) (45000,0.23) (46000,0.24) (47000,0.249) (48000,0.259) (49000,0.269) (50000,0.28)};
\addplot coordinates { (20000,0.062) (21000,0.065) (22000,0.071) (23000,0.076) (24000,0.081) (25000,0.086) (26000,0.092) (27000,0.098) (28000,0.105) (29000,0.11) (30000,0.117) (31000,0.124) (32000,0.131) (33000,0.138) (34000,0.146) (35000,0.154) (36000,0.161) (37000,0.169) (38000,0.177) (39000,0.185) (40000,0.194) (41000,0.202) (42000,0.211) (43000,0.22) (44000,0.23) (45000,0.24) (46000,0.25) (47000,0.26) (48000,0.27) (49000,0.28) (50000,0.29)};
\pgfplotsset{cycle list shift=4}
\addplot coordinates { (20000,0.056) (21000,0.06) (22000,0.065) (23000,0.07) (24000,0.075) (25000,0.081) (26000,0.087) (27000,0.093) (28000,0.099) (29000,0.105) (30000,0.112) (31000,0.119) (32000,0.126) (33000,0.133) (34000,0.14) (35000,0.147) (36000,0.155) (37000,0.163) (38000,0.171) (39000,0.18) (40000,0.188) (41000,0.197) (42000,0.205) (43000,0.215) (44000,0.224) (45000,0.234) (46000,0.244) (47000,0.254) (48000,0.263) (49000,0.274) (50000,0.284)};
\addplot coordinates { (20000,0.062) (21000,0.067) (22000,0.072) (23000,0.077) (24000,0.083) (25000,0.088) (26000,0.094) (27000,0.1) (28000,0.106) (29000,0.113) (30000,0.12) (31000,0.127) (32000,0.134) (33000,0.142) (34000,0.149) (35000,0.157) (36000,0.165) (37000,0.173) (38000,0.181) (39000,0.189) (40000,0.198) (41000,0.207) (42000,0.217) (43000,0.225) (44000,0.236) (45000,0.244) (46000,0.256) (47000,0.265) (48000,0.276) (49000,0.285) (50000,0.297)};
\addplot coordinates { (20000,0.042) (21000,0.046) (22000,0.049) (23000,0.054) (24000,0.056) (25000,0.06) (26000,0.063) (27000,0.068) (28000,0.072) (29000,0.076) (30000,0.08) (31000,0.086) (32000,0.089) (33000,0.094) (34000,0.099) (35000,0.104) (36000,0.109) (37000,0.113) (38000,0.12) (39000,0.124) (40000,0.129) (41000,0.134) (42000,0.143) (43000,0.146) (44000,0.152) (45000,0.16) (46000,0.164) (47000,0.17) (48000,0.176) (49000,0.183) (50000,0.189)};
\addplot coordinates { (20000,0.05) (21000,0.052) (22000,0.055) (23000,0.061) (24000,0.063) (25000,0.067) (26000,0.071) (27000,0.074) (28000,0.079) (29000,0.082) (30000,0.087) (31000,0.091) (32000,0.097) (33000,0.103) (34000,0.106) (35000,0.111) (36000,0.117) (37000,0.123) (38000,0.128) (39000,0.133) (40000,0.138) (41000,0.143) (42000,0.15) (43000,0.156) (44000,0.162) (45000,0.167) (46000,0.174) (47000,0.181) (48000,0.186) (49000,0.193) (50000,0.2)};

\legend{\df,\mdf,\vdh,\mvdh,\ebd,\mebd}
\end{axis}

\end{tikzpicture}
\caption{Running times on a series of self-additions as depicted in Figure~\ref{fig:graphweights}. Counting operators without removing duplicates does not improve on the default running time. The depth heuristic reduces the default running time by up to \SI{32}{\percent}.}
\label{fig:commonsubexp}
\end{figure}
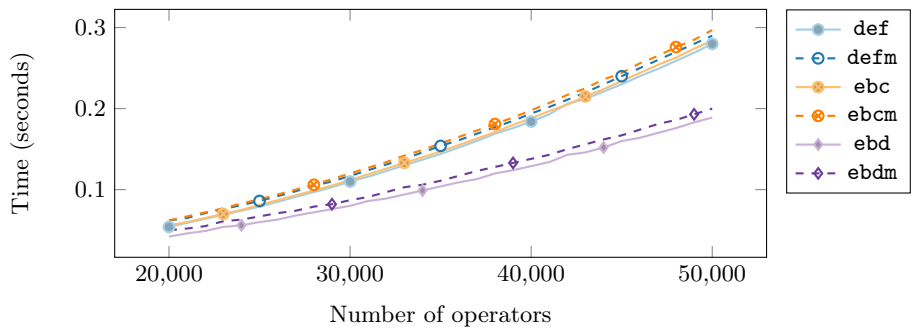

If common subexpressions lead to a large difference between the actual number of operators and the number of operators in a tree expansion, \vdh significantly overestimates the optimal weight of its edges (cf.\ Figure~\ref{fig:graphweights}). Figure~\ref{fig:commonsubexp} shows results for evaluating a sequence of additions where the left and the right summand is the result of the previous addition. The full operator count heuristic does not reduce the running time and even performs worse than the default strategy for large numbers of operators, whereas the depth-based heuristic clearly outperforms the other strategies. Note that in this case, the depth-based heuristic leads to the optimal error distribution for both total and critical path cost.

\subsection{A note on floating-point primitives}

Error bound balancing requires the use of floating-point error bounds. While \ieee requires that floating point computations must be exactly rounded, it is surprisingly difficult to find an adequate upper or lower bound to the result of such an operation. \ieee specifies four rounding modes: Round to nearest, Round to positive/negative infinty and Round to zero~\cite{ieee2008}. For the last three modes, which are commonly referred to as directed rounding, it is easy to obtain a lower or upper bound by negating the operands adequately. Unfortunately, most systems implement Round to nearest. Switching the rounding mode is expensive. While \double operations with appropriate negations for directed rounding are about two times slower, switching to an appropriate rounding mode can increase the running time of a single operation by a factor of $100$. The same factor applies if we manually jump to the next (or previous) representable \double value.

Handling floating-point primitives correctly can, depending on the architecture, be very expensive. In most cases, however, the computed error bounds massively overestimate the actual error. Moreover, for the actual bigfloats computations the error bounds are rounded up to the next integer. It is therefore almost impossible that floating-point rounding errors make an actual difference in any computation. For our experiments we refrained from handling those bounds correctly to make the results more meaningful and less architecture-dependent. 


\section{Conclusion}

We have shown, theoretically and experimentally, that both external and internal balancing methods are useful tools to mitigate the impact of badly balanced expression dags. Restructuring has a higher potential on reducing the cost, but can become useless or even detrimental if the graph has many common subexpressions or is already balanced. In a parallel environment, restructuring is necessary to make use of multiple processors in an unbalanced graph. Error bound balancing is more widely applicable, but is limited in its effectivity. If the graph is small or already sufficiently balanced, neither of the methods has a significant positive impact on the evaluation cost. A general purpose number type should therefore always check whether the structure generally requires balancing before applying either of the algorithms.
For both strategies we have described optimal weight functions. In both cases implementations require heuristics to be practicable. Our experiments show that carefully chosen heuristics are in most cases sufficient to increase the performance of exact number types.

\bibliographystyle{splncs04}
\bibliography{lit.bib}

\end{document}